\documentclass[a4paper]{article}
\usepackage[utf8]{inputenc}
\usepackage{amsthm}
\usepackage{amsmath}
\usepackage{amssymb}
\usepackage[numbers,sort&compress]{natbib} 
\usepackage[breaklinks,pagebackref]{hyperref}
\usepackage{booktabs}
\usepackage{enumerate}
\usepackage{paralist}
\usepackage{psfrag}
\usepackage{graphicx}
\usepackage{tikz}
\usetikzlibrary{calc}
\usepackage[font=small]{subfig}
\usepackage[disable,textsize=scriptsize]{todonotes}

\newtheorem{theorem}{Theorem}[section]
\newtheorem{lemma}{Lemma}[section]

\DeclareMathOperator{\bigO}{O}

\newcommand{\dsg}{\textsc{Densest-Subgraph}}

\author{Manuel Sorge}

\title{A More Complicated Hardness Proof for Finding Densest Subgraphs in Bounded Degree Graphs}

\begin{document}

\maketitle
\begin{abstract}
  We consider the \dsg{} problem, where a graph and an integer~$k$ is given and we search for a subgraph on exactly~$k$ vertices that induces the maximum number of edges. We prove that this problem is NP-hard even when the input graph has maximum degree three.
\end{abstract}

\section{Introduction}
\label{sec:intro}
We consider the following problem.
\begin{quote}
  \dsg :\\
  \textbf{Input:} A graph~$G=(V,E)$, and a nonnegative integer~$k$.\\
  \textbf{Task:} Find a vertex set~$S\subseteq V$ of size exactly~$k$ such
  that~$G[S]$ has maximum number of edges.
\end{quote}
We call every vertex set~$K \subseteq V$ with~$|K|=k$ a \emph{solution} for this particular instance. \dsg{} is clearly a fundamental problem and has received much attention in the literature~\cite{Kos04, II09, BBH11, FPK01}. In particular, this problem has been proven NP-hard in various variants. We prove that \dsg{} remains NP-hard on graphs of degree at most three. While this result has been proven before~\cite{FS97}, we show it in a much more complicated manner. More precisely, we first give a reduction from {\sc Clique} to \dsg{} with maximum degree five, proving its correctness through an intricate replacement argument. Then we successively reduce the maximum degree by replacing each high-degree vertex by gadget graphs. We prove that, for some optimal solution~$K$, each gadget graph is either fully contained in~$S$ or not at all. This is done using elaborate case-distinctions which are most tedious to verify. We note that the previous proof is essentially contained within two pages. Our proof needs eight pages.

\paragraph{Preliminaries.} We use standard graph notation as used, for example, by Diestel~\cite{Die10}. Where it is appropriate, we denote~$n := |V(G)|$ and~$m := |E(G)|$.

\section{NP-hardness of \dsg{} in Graphs with Maximum Degree Three}
Our proof of NP-hardness is divided into three steps. We first show that \dsg{} is NP-hard in graphs with maximum degree at most five. Then, we proceed to reduce the maximum degree in hard instances. We successively employ two gadgets that replace vertices in the input graph, giving hardness in graphs with degree at most four and then hardness in graphs with degree at most three.
\begin{lemma}
  \dsg{} is NP-hard even on graphs with maximum degree five.
\end{lemma}
\begin{proof}

  Let~$(G = (V, E), s)$ be an instance of~\textsc{Clique} and without loss of generality assume that $|V|=n$~is an odd square number. The construction is as follows. For every vertex~$v$ in~$V$ introduce a quadratic grid graph~$T_v$ with $(n^2 + 1) \cdot (n^2 + 1)$~vertices. In order to simplify the analysis, identify vertices with degree three on opposite boundaries, creating a grid model of a torus in which every vertex has degree exactly four. More formally,~$T_v$ is the graph with vertex set~\[\left\{v^i_j \colon 0 \leq i, j, \leq n^2 - 1\right\}\] and edge set~\[\left\{\left\{v^i_j, v^{(i + 1) \bmod n^2}_j\right\}, \left\{v^i_j, v^i_{(j + 1) \bmod n^2}\right\} \colon 0 \leq i, j \leq n^2 - 1\right\}\text{.}\]
  Let~$\phi \colon V \to \{0, \ldots, n - 1\}$ be a bijection and let
  \begin{align*}
    \psi &\colon V \to \mathbb{N} \colon i \mapsto \left(\phi(i)n\sqrt{n}\right) \bmod n^2 \text{, and}\\
    \sigma &\colon V \to  \mathbb{N} \colon i \mapsto \left\lfloor \phi(i)n\sqrt{n} / n^2 \right\rfloor\text{.}
  \end{align*}
  For every edge~$\{u, v\} \in E$ connect the tori~$T_u$ and~$T_v$ by an ``inter-torus'' edge between the vertices~$u^{\psi(v)}_{\sigma(v)}$ and $v^{\psi(u)}_{\sigma(u)}$. Call the graph that is constructed in this way~$G^*$ and set the instance of \dsg{} \todo{{\sc Densest-$sn^4$-Subgraph}!?}to~$(G^*, sn^{4})$. Clearly, this construction can be carried out in polynomial time. We decide~$(G, s)$ to be a yes-instance if and only if an optimal solution of~$(G^*, sn^{4})$ contains at least $2sn^{4} + s(s-1)/2$~edges.

  For the correctness, first, observe that~$G$ contains a subgraph of~$G^*$ with $2sn^{4} + s(s-1)/2$~edges if~$G$ contains a clique with $s$~vertices: for every vertex~$v$ in the clique choose every vertex of the torus~$T_v$. For the other direction, we prove that every optimal solution to~$(G^*, sn^{4})$ comprises a collection of complete tori. That is, if an optimal solution contains a vertex of one of the tori~$T_v$, every vertex of~$T_v$ is in the solution. Then it is clear that if the optimal solution contains $2sn^{4} + s(s-1)/2$~edges, the vertices in~$G$ that correspond to a torus~$T_v$ in the optimal solution must induce a clique on $s$~vertices.

  An optimal solution~$K$ to~$(G^*, sn^{4})$ consists of a number of possibly proper vertex subsets of the tori, altogether containing exactly~$sn^{4}$~vertices. For the sake of contradiction assume that at least one of these vertex subsets is proper, that is, not all of the vertices of some torus are contained in~$K$. Call the vertices in~$K$ ``black'' and the vertices in~$V \setminus K$ ``white''. Call edges between black and white vertices within one torus ``cut''. Call tori~$T_v$ with at most~$n^4 - n^3$ black vertices ``small'' and the remaining tori ``large''. We prove the following.
  \begin{quote}
    \textbf{Claim:} In each small torus~$T_v$, there are at least twice as many cut edges as there are inter-torus edges incident with black vertices of~$T_v$. 
  \end{quote}
  Before proving the claim, we note that it implies the lemma. There are at most~$sn^4/(n^4 - n^3) = sn/(n-1)$ large tori, that is, for every sufficiently large\footnote{The relation $sn/(n-1) \geq s + 1$ holds if and only if~$s/(s+1) \geq (n-1)/n$. Since~$(n - 1)/n$ is strictly monotone ascending for~$n \geq 1$ and equality holds for~$n = s + 1$, we have that~$sn/(n-1) < s + 1$ for all~$n > s + 1$. We can assume this without loss of generality.}~$n$, there are at most~$s$ large tori. Consider distributing all black vertices of small tori to large tori, completing them in the process, and combining the remaining black vertices in small tori to complete tori. This procedure cannot decrease the number of edges induced by the black vertices: All inter-torus edges between large tori are preserved and for every inter-torus edge incident with a black vertex of a small torus, there were at least two cut edges, and thus, at least one edge is added to the solution. (Observe that after this procedure, there are no cut edges remaining.) Furthermore, after this procedure, each black vertex has at least four black neighbors within its torus, whereas before it had at most four. Thus, we may assume that there are no incomplete tori and the correctness of our construction follows.

  We prove our claim using the following result by Efe and Feng~\cite{EF07}.
  \begin{quote}
    \textbf{Fact:} In order to remove $x$~vertices from a quadratic grid with $N \cdot N$~vertices, where $N$~is odd, at least~$2\min\{x, N^2 - x\}/(N - 1)$ edges have to be cut.
  \end{quote}
  It is clear that in a torus with the same number of vertices at least as many edges have to be cut.\footnote{Consider removing a set of vertices from a quadratic grid and then removing it from a torus on the same vertices as the grid. The number of cut edges incident with each vertex differ only for vertices that are removed from the boundary of the grid. But in the torus, these vertices have higher degree than in the grid.} From this statement, our claim follows directly for tori with at least~$n^3$ and at most~$n^4-n^3$ vertices. Assume that there is a torus~$T_v$ with~$n^4/2 \leq n_v \leq n^4 - n^3$ black vertices. The number of cut edges is
  \begin{align*}
    \frac{2(n^4 - n_v)}{n^2 -1} \geq \frac{2n^3}{n^2 - 1} > 2n\text{.}
  \end{align*}
  Since there are at most~$n$ inter-torus edges incident with~$T_v$, there are at least twice as many cut edges as there are inter-torus edges. The same follows analogously for tori~$T_v$ with~$n^3 \leq n_v \leq n^4/2$ black vertices.

  It remains to show the claim for tori with at most~$n^3$ black vertices. Consider such a torus~$T_v$ and call a vertex~$v^i_j$ in~$T_v$ to reside in ``row''~$i$ and in ``column''~$j$. We show that incident with each connected component~$C$ of black vertices in~$T_v$ there are at least twice as many cut edges as there are inter-torus edges incident with black vertices in~$C$. The claim then follows, since there are no cut edges that are incident with two of~$T_v$'s black connected components.  First, assume that~$C$ contains each vertex of some row~$i$. Since~$T_v$ contains at most~$n^3$ black vertices, the number of columns in~$T_v$ with fewer than~$n^2$ black vertices is at least~$n^2 - n$ and this gives a lower bound on the number of cut edges incident with~$C$. Thus, since there are at most~$n$ inter-torus edges incident with~$T_v$, for every sufficiently large~$n$, there are at least twice as many cut edges as inter-torus edges incident with~$C$. Analogously we can prove this, if~$C$ contains each vertex of some column. Next, assume that~$C$ does not contain complete rows or columns. If there is only one inter-torus edge incident with~$C$, clearly there are always more than two cut edges. If there are at least two inter-torus edges incident with~$C$, consider a shortest black path between these two edges. By the placement of the inter-torus edges in~$T_v$, this path either touches at least~$n\sqrt{n}$ rows or at least as many columns. Since there are no complete columns or rows in~$C$,~$n\sqrt{n}$ is then a lower bound on the number of cut edges incident with~$C$. Thus, our claim now follows for all~$n$ larger than some constant and our construction is correct.
\end{proof}
Next, we reduce the maximum degree in hard instances to four, by replacing each vertex with a ``solid'' gadget, that is, a gadget that can be assumed to be either completely contained in a solution or to be disjoint to it. Thus, in order to maximize the number of edges induced by a solution, one has to choose gadgets correspondingly to vertices in the original graph.
\begin{lemma}\label{lem:dks-md4}
  There is a polynomial-time many-one reduction from \dsg{} in graphs with maximum degree six or maximum degree five to \dsg{} in graphs with maximum degree four.
\end{lemma}
\begin{proof}
  Let~$(G, k)$ be an instance of \dsg{} where each vertex in~$G$ has degree at most six or at most five. We replace each vertex~$v$ in~$G$ by the graph shown in \autoref{fig:dks-fence-gadget} (the ``fence gadget'') and distribute~$v$'s edges to the outer vertices such that each vertex in the fence gadget gets degree at most four. Call the graph constructed in this way~$G^*$ and set the new instance of \dsg{} \todo{{\sc Densest-$8k$-Subgraph}!?}to~$(G^*, 8k)$. We argue that from any solution~$K$ to~$(G^*, 8k)$ we can construct another solution with at least as many edges such that each fence gadget either is completely contained in~$K$ or its vertices are disjoint to~$K$. Then, from any solution to~$(G^*, 8k)$ with~$13k + x$ edges (there are~$13$ edges in a gadget) we can construct a  solution for~$(G, k)$ with at least~$x$ edges and vice-versa by completing every gadget and then exchanging gadgets with their corresponding vertices. Thus, since it is easy to achieve at least~$13k$ edges in a solution for~$(G^*, 8k)$, finding an optimal solution or any solution with a number of edges above a given threshold is equivalent in these instances. It is also not hard to see, that constructing the solution for~$(G, k)$ from a solution for~$(G^*, 8k)$ can be done in polynomial time, we omit the details.\todo{MS: Dürfen wir das? Das ist: höchstens~$k$-mal iterieren zum Komplettieren der Gadgets. Darin: Immer wieder den Claim herstellen (mittels iterieren über alle Gadgets und Property 1 und 2 herstellen = $\bigO(n)$), und beim Komplettieren entsprechende Austauschknoten finden = auf jeden Fall $\bigO(n^8)$. Wäre unschön, den Text noch mehr vollzustopfen. Oder einen Extraabsatz ganz am Ende dazu?}
  \begin{figure}[t]
    \centering
    \begin{tikzpicture}
      \tikzstyle{every node} = [circle, fill=gray!30]
      \node (v7) at (0: 0) {$v_7$};
      \node (v2) at +(90: 1) {$v_2$};
      \node (v1) at +(2 * 90: 1) {$v_1$};
      \node (v6) at +(3 * 90: 1) {$v_6$};
      \node (v8) at (0: 1) {$v_8$};
      \node (v3) at ($(0: 1)+(90: 1)$) {$v_3$};
      \node (v4) at ($(0: 1)+(0: 1)$) {$v_4$};
      \node (v5) at ($(0: 1)+(- 90: 1)$) {$v_5$};

      \draw (v1) -- (v2)
            (v1) -- (v7)
            (v1) -- (v6)
            (v2) -- (v7)
            (v2) -- (v3)
            (v3) -- (v8)
            (v3) -- (v4)
            (v4) -- (v5)
            (v4) -- (v8)
            (v5) -- (v6)
            (v5) -- (v8)
            (v6) -- (v7)
            (v7) -- (v8);

    \end{tikzpicture}
    \caption{The fence gadget used in \autoref{lem:dks-md4}. Every vertex~$v$ in the original graph is replaced by this graph. The edges of~$v$ are distributed to the \emph{outer} vertices~$v_1, \ldots, v_6$. The \emph{inner} vertices~$v_7, v_8$ do not get additional edges.}
    \label{fig:dks-fence-gadget}
  \end{figure}
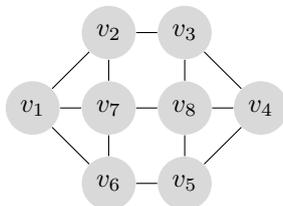
   
  Let~$K$ be a solution for~$(G^*, 8k)$. Call the vertices in~$K$ ``black'' and the remaining ones ``white''. The \emph{black-degree} of a vertex is the number of black vertices adjacent to it. Call edges between two black vertices ``black'' and the remaining edges ``non-black''. We consider fence gadgets that are partially black and argue that we can move black vertices between gadgets without losing edges, ultimately completing all gadgets. To do this, we prove the following.
  \begin{quote}
    {\bf Claim:} Given a solution~$K$, we can modify~$K$ without losing black edges and without changing the number of black vertices in any fence gadget such that the following holds.
    \begin{enumerate}
    \item For~$i \in \{1, 2, 3\}$ it holds that in each fence gadget with exactly~$i$ white vertices, there are at least~$2i + 1$ non-black edges.
    \item For~$j \in \{1, 2, 3\}$ it holds that in each fence gadget with at least~$j$ black vertices and at least~$j$ white vertices, there is a set of~$j$ black vertices in the gadget that is incident with at most~$2j + 1$ black edges (counting both inner-gadget edges and edges between gadgets).
    \item Each fence gadget with exactly four white vertices has at least eight non-black edges (within the gadget) and a set of four black vertices that is incident with at most eight black edges (both within and outside of the gadget).
    \item Each fence gadget with exactly~$1 \leq i \leq 3$ black vertices contains a black vertex with black-degree at most two.
    \end{enumerate}
  \end{quote}
  \todo{MS: Ich finde es besser, erst zu zeigen, wie man den Claim anwendet und ihn dann zu beweisen. Egal wierum man es macht, man muss immer scrollen; entweder ist die Figure mit dem Gadget nicht da, oder der Claim selbst, wenn man ihn anwendet. Aber bei der Variante den Claim erst hinterher zu beweisen, verliert man imho weniger Motivation beim Beweislesen, weil man erst sieht, warum der Claim so aussieht, wie er aussieht, und nicht erst irgendeine Aussage überprüfen muss, von der man keine Ahnung hat, wie sie angewendet wird oder einen weiterbringt.}
  We first show how this claim can be used to complete all gadgets in~$K$. Assume that the claim holds. We show how we can make gadgets completely black that have some minimum number of white vertices other than~$8$. Let~$1 \leq i < 8$ be the minimum number of white vertices in gadgets and let~$g_i$ be a gadget with exactly~$i$ white vertices.
  \begin{itemize}
  \item Assume that~$i = 1$. It follows that~$g_i$ contains a white vertex~$v$ that has at least three black neighbors and Part~2 of our claim tells us that we can color a vertex in another gadget white and color~$v$ black without losing black edges. 
  \item Assume that~$i = 2$. 
    Then, either there are two gadgets~$g', g''$ with exactly one black vertex or there is a gadget~$g'$ with at least two black vertices and at least two white vertices. In the first case, we can color both black vertices of~$g'$ and~$g''$ white, losing at most two black edges in the process, and complete~$g_i$, gaining at least five black edges (Part~1 of the claim). In the second case, we use Part~2 to color two black vertices of~$g''$ white, losing at most five black edges and complete~$g_i$, gaining at least five black edges. 
  \item Assume that~$i = 3$. Then, either we can use Part~2 or Part~4 of the claim to find three black vertices such that, if we color them white, we lose at most seven black edges. We color the found vertices white and complete the gadget~$g_i$ without losing black edges by Part~1.
  \item Assume that~$i = 4$. Then, either we can find another gadget with exactly four white vertices and use Part~3 of our claim to color its vertices white and complete~$g_i$ without losing black edges, or, if there is no other gadget with exactly four white vertices, by Part~4 of our claim, we can successively find four vertices each with black-degree at most two, color them white and, after that, we can complete~$g_i$. We lose at most eight black edges, but gain at least eight by Part~3 of the claim.
  \item Assume that~$i = 5$. Observe that in a gadget with exactly three black vertices, the number of non-black edges is at least ten. Thus, we can again apply Part~4 of our claim. 
  \item The case that~$i = 6$ is analogous to~$i = 5$ and the case~$i = 7$ is trivial. 
  \end{itemize}
  To make each gadget either completely black or completely white, we first ensure that each part of the claim holds, then apply a modification as described above to complete the gadget~$g_i$ and iterate. The modifications that are necessary for each part of the claim to hold do not change the number of black vertices in any gadget, and in each of the above cases, we modify~$K$ such that the number of completely black gadgets increases. Thus, after at most~$k$ steps, each gadget is either completely black or completely white, and the correctness of our construction follows if the claim holds.

  Next, we prove the claim. We first prove that we can make two assumptions about the solution~$K$ without loss of generality. We modify~$K$ without losing edges and without changing the number of black vertices in any fence gadget such that the following holds.
  \begin{quote}
    {\bf Property 1:} In every gadget, there is no white inner vertex that has at most one outer white neighbor and no white inner neighbor.
  \end{quote}
  Assume that there is such a white inner vertex~$v$. If~$v$ does not have a white outer neighbor, we may simply color one of its black outer neighbors white, losing at most three black edges and then color~$v$ black, gaining at least three black edges. If~$v$ has a white outer neighbor, then it has a black outer neighbor~$w$ with black-degree at most two. Thus, we can make~$w$ white, losing at most two black edges, and make~$v$ black, gaining at least two black edges.
  \begin{quote}
    {\bf Property 2:} In every gadget with at most four white vertices, the white vertices induce a connected graph.
  \end{quote}
  Consider a fence gadget with at least two white vertices. If there is a singleton white vertex~$w$, we color it black and color white a black neighbor~$b$ of one of the other white vertices. In the exchange, if~$b$ and~$w$ are not neighbors, at most three black edges incident with~$b$ are lost and at least three black edges incident with~$w$ are gained. If~$b$ and~$w$ are neighbors, at most two black edges are lost and at least two black edges are gained. In the following, assume that there are no singleton white vertices. If there are exactly three white vertices that do not induce a connected graph, then there is at least one singleton white vertex. Hence, we may assume that there are at least four white vertices. We remove each maximal connected component~$c$ with at most two white vertices: Observe that if both inner vertices are white, the white vertices always induce a connected graph. Thus, we may assume that there is only one inner white vertex and, by Property~1,~$c$ has no inner vertices. We now know that if the white vertices do not induce a connected graph, there are at least four of them, the minimum size of a white connected component is two, each white connected component of size exactly two contains no inner vertices, and there is at most one white inner vertex. 
  Without loss of generality, there are only two 
  possible configurations left. Either the 
  vertices~$v_1,v_2, v_4, v_5$ or the vertices~$v_2, v_3, v_5, v_6$ are white. 
  In case one, 
  we can make~$v_1$ black and~$v_3$ white, losing at most two black edges and gaining at least two black edges. In case two, we can color~$v_2$ black and~$v_4$ white. Property~2 now follows. 

  Notice that ensuring either property does not change the number of black vertices in the gadget. Furthermore, since ensuring Property~1 does not destroy Property~2, and since Property~1 does not depend on Property~2, we can modify the solution~$K$ such that both statements hold simultaneously.

  To prove Part~1 through~4 of the claim, we iterate through the the number~$i$ of white vertices in a gadget, and show that each of the relevant statements hold; no further modifications are necessary. Let~$g_i$ be a gadget with exactly~$i$ white vertices.
  \begin{itemize}
  \item The claim is trivial for~$i = 1$.
  \item Assume that~$i = 2$. Both white vertices of~$g_i$ have at least three incident non-black edges. The neighborhood of the white vertices overlaps in at most one vertex, and, thus, we have at least~$2 \cdot 3-1$ non-black edges, matching the given bound of Part~1 of the claim. Part~2 of the claim follows by choosing two black vertices~$v, w$ in the neighborhood of the white vertices such that~$v, w$ are neighbors. It is easily seen that such black vertices must exist. 
  \item Assume that~$i = 3$. 
    \begin{itemize}
    \item If both inner vertices of~$g_i$ are white, these two vertices alone have seven incident non-black edges and Part~1 of the claim follows. For Part~2 and~$j = 3$, we choose three consecutive black vertices on the cycle formed by the outer vertices. Each of these vertices has at most three incident black edges and the neighborhoods overlap in at most one vertex for neighboring vertices. Thus, these three vertices have at most~$3 \cdot 3 - 2 = 7$ incident black edges. If we are to choose only two black vertices, that is,~$j = 2$, we elect the neighbors of the outer white vertex. Both these black vertices have black-degree at most two. The case~$j = 1$ is trivial and Part~2 of the claim follows. 
    \item If exactly one inner vertex is white, then, by Property~1, it has two white outer neighbors. Without loss of generality, 
      let the white vertices be~$v_3, v_4, v_8$. These three vertices have at least seven incident non-black edges and the first part of the claim follows. Choosing the black vertices~$\{v_2, v_5, v_7\}$, $\{v_2, v_7\}$, and~$\{v_5\}$ yields at most~$7$, $5$, and~$2$ incident black edges, respectively. Thus, Part~2 of the claim follows. 
    \item Now, if all three white vertices are outer vertices, by Property~2, they form a connected component. Without loss of generality, 
      there are only two possible configurations:~$v_2, v_3, v_4$, and~$v_3, v_4, v_5$, each yielding at least seven non-black edges, fulfilling the first part of the claim. For Part~2, we choose~$\{v_5, v_7, v_8\}$, $\{v_7, v_8\}$, and~$\{v_8\}$, yielding at most~$6$, $4$, and~$2$ incident black edges in the first configuration. For the second configuration, we choose~$\{v_2, v_7, v_8\}$, $\{v_7, v_8\}$, and~$\{v_8\}$, yielding at most~$6$, $4$, and~$1$ incident black edges and, thus, Part~2 of the claim holds. 
    \end{itemize}
  \item Now assume that~$i = 4$. 
    \begin{itemize}
    \item If both inner vertices are white, they yield at least seven non-black vertices plus at least two from one of the outer white vertices, satisfying the lower bound of Part~3 of the claim. For the upper bound on the number of black edges incident with some black vertices in Part~3, observe that there is a black vertex with black-degree at most two (incident with one of the outer white vertices) and deleting it makes its black neighbor (if there is any) have black-degree at most two. For the upper bound from Part~2 of our claim, analogously to~$i = 3$, we choose three consecutive black vertices or, if this is not possible, neighbors of the white outer vertices. It is easy to check that Part~2 of the claim holds in this case. 
    \item Now assume that exactly one inner vertex is white. Then, by Property~1, it has at least two white outer neighbors. Without loss of generality, 
      the only possible configurations are $v_2, v_3, v_4, v_8$ and $v_3, v_4, v_5, v_8$. Each configuration yields at least eight non-black edges. Since there are exactly~$13$ edges in the gadget and at most three black vertices with edges, the upper bound of Part~3 of our claim follows. For the upper bound of Part~2, we choose $\{v_1, v_5, v_7\}$, $\{v_7, v_5\}$, and~$\{v_7\}$ in configuration one, yielding at most~$6$, $4$, and~$2$ black edges, and $\{v_2, v_6, v_7\}$, $\{v_2, v_7\}$, and~$\{v_2\}$, yielding at most $7$, $5$, and~$3$ edges. 
    \item Assume that every white vertex is an outer vertex. Then, by Property~2, they induce a connected graph. Thus, without loss of generality, the only possible configurations are~$v_1$ through~$v_4$ and~$v_2$ through~$v_5$. The lower bound on the number of non-black edges in Part~3 follows, since the vertices~$v_2, v_3, v_4$ are incident with seven edges in the gadget and both~$v_1$ and~$v_5$ contribute at least one further non-black edge. The upper bound on black edges from Part~3 follows, since there are at most~$13 - 8 = 5$ black edges within the gadget and at most two black edges not within the gadget incident to~$v_6$ and~$v_5$ or~$v_1$, respectively. For Part~2 we may choose the vertices~$\{v_6, v_7, v_8\}$, $\{v_7, v_8\}$, and~$\{v_6\}$, yielding at most~$5$, $3$, and~$3$ incident black edges in both configurations. 
    \end{itemize}
  \item It remains to prove Part~2 and Part~4 of our claim for~$5 \leq i \leq 7$. The only nontrivial case is~$i = 5$, that is, there are exactly three black vertices. Both Part~2 and Part~4 trivially follow, if the black vertices are not connected. If they are connected, they either form a triangle or a path. In the first case, at most two black vertices are incident with non-gadget edges, and, thus, the remaining vertex has black-degree at most two and the black vertices have at most five incident black edges. In the second case, there are two black vertices with black-degree at most two and the black vertices have incident at most five black edges. 
  \end{itemize}
  Thus, our claim follows, and our construction is correct.
\end{proof}
We now use the same strategy, albeit with a much simpler proof, replacing the vertices by a further gadget to further reduce the maximum degree in the input graph.
\begin{lemma}
  There is a polynomial-time many-one reduction from \dsg{} in graphs with maximum degree four to \dsg{} in graphs with maximum degree three.
\end{lemma}
\begin{proof}
  Let~$(G, k)$ be an instance of \dsg{} where each vertex in~$G$ has degree at most four. We replace each vertex~$v$ in the graph~$G$ by a four-vertex cycle (the ``cycle gadget''), and distribute~$v$'s edges to the vertices of the cycle such that each vertex in the cycle gets degree at most three. Call the graph constructed in this way~$G^*$. We set the new instance of \dsg{} to~$(G^*, 4k)$. We show \todo{MS: Was spricht gegen argue?}that this construction is correct by showing that from any solution to~$(G^*, 4k)$ we can construct a solution with at least as many edges such that the vertices of each cycle gadget are either completely contained in the solution or are disjoint. Thus, from any solution to~$(G^*, 4k)$ with~$4k + x$ edges, we can construct a solution for~$(G, k)$ with at least~$x$ edges and vice-versa. It is clear, that solutions for~$(G*, 4k)$ with at least~$4k$ edges are achievable in polynomial time, and it will also not be hard to check that we construct the solution for~$(G, k)$ from a solution for~$(G^*, 4k)$ in polynomial time. Thus, finding an optimal solution or any solution with a number of edges above some given threshold is polynomial-time equivalent in these instances.

  Consider a solution~$K$ to~$(G^*, 4k)$. We move vertices between gadgets that are not fully contained in~$K$, and ultimately achieve that each gadget is either completely contained in~$K$ or disjoint to it. Furthermore, in the process the number of edges induced by~$K$ does not decrease. We proceed step-wise, completing at least one gadget in each step. Observe first, that in each step, we may assume that incomplete gadgets contain at least two vertices from~$K$. Otherwise, we may simply move the vertices of~$K$ from one gadget to another without losing edges induced by~$K$. Now, consider a cycle gadget~$g$ containing exactly three vertices of~$K$, if there is any. Every incomplete cycle gadget has at least one vertex with degree at most two in~$G^*[K]$. Thus, since there is at least one further incomplete cycle gadget, we may simply move such a low degree vertex to~$g$ without losing edges. If there are no cycle gadgets with exactly three vertices, there remain only cycle gadgets with exactly two vertices in~$K$. However, since the vertices of these gadgets are incident with at most three edges in~$G^*[K]$, we can move two vertices from one gadget to another without losing edges. This completes the description of one step. In each step, at least one gadget is completed, and after at most~$k$ steps of moving vertices of~$K$ between gadgets, we obtain a solution~$K$ such that every cycle gadget is either completely contained in~$K$ is disjoint to it.
\end{proof}

Combining the three lemmas of this section, the following is now easily obtained.
\begin{theorem}
  \label{thm:dsg-three}
  \dsg{} is NP-hard even on graphs with maximum degree three.
\end{theorem}

\section{Conclusion and Outlook}
\label{sec:conclusion}

We have demonstrated that stubbornness not only may lead to the proof of a desired result. It also can be used to obtain much more complicated proofs of previously known ones. We leave it as an open question as to which elegant proofs can be reproved more elaborately and in a nontrivial way.

\paragraph{Acknowledgement.} I thank Christian Komusiewicz for insightful discussions about the matter of the paper and for proofreading.

\bibliographystyle{abbrv}
\bibliography{mu-clique,more-comp}

\begin{thebibliography}{1}

\bibitem{BBH11}
B.~Balasundaram, S.~Butenko, and I.~V. Hicks.
\newblock Clique relaxations in social network analysis: The maximum $k$-plex
  problem.
\newblock {\em Oper. Res.}, 59(1):133--142, 2011.

\bibitem{Die10}
R.~Diestel.
\newblock {\em Graph Theory}, volume 173 of {\em Graduate Texts in
  Mathematics}.
\newblock Springer Verlag, New York, 4th edition, 2010.

\bibitem{EF07}
K.~Efe and G.-L. Feng.
\newblock A proof for bisection width of grids.
\newblock {\em International Journal of Mathematical and Computer Sciences},
  4(3), 2008.

\bibitem{FPK01}
U.~Feige, D.~Peleg, and G.~Kortsarz.
\newblock The dense {\it k}-subgraph problem.
\newblock {\em Algorithmica}, 29(3):410--421, 2001.

\bibitem{FS97}
U.~Feige and M.~Seltser.
\newblock On the densest $k$-subgraph problem.
\newblock Technical report, The Weizmann Institute, Department of Applied Math
  and Computer Science, 1997.

\bibitem{II09}
H.~Ito and K.~Iwama.
\newblock Enumeration of isolated cliques and pseudo-cliques.
\newblock {\em ACM Transactions on Algorithms}, 5(4), 2009.

\bibitem{Kos04}
S.~Kosub.
\newblock Local density.
\newblock In {\em Network Analysis}, volume 3418 of {\em LNCS}, pages 112--142.
  Springer, 2004.

\end{thebibliography}

\end{document}